\documentclass[onecolumn,a4paper,11pt,accepted=2021-03-18]{quantumarticle}
\pdfoutput=1
\usepackage[utf8]{inputenc}
\usepackage[english]{babel}
\usepackage[T1]{fontenc}
\usepackage{hyperref}

\usepackage{amsmath,amssymb,amsthm,mathtools}

\usepackage{mleftright}\mleftright
\usepackage{float}

\usepackage{enumitem}

\usepackage{xcolor}

\def\01{\{0,1\}}
\newcommand{\ceil}[1]{\lceil{#1}\rceil}

\newcommand{\eps}{\varepsilon}
\newcommand{\ket}[1]{|#1\rangle}
\newcommand{\bra}[1]{\langle#1|}
\newcommand{\ketbra}[2]{|#1\rangle\langle#2|}
\newcommand{\braket}[2]{\langle#1|#2\rangle} 
\newcommand{\Tr}{\mbox{\rm Tr}}
\newcommand{\norm}[1]{\mbox{$\parallel{#1}\parallel$}}

\newtheorem{theorem}{Theorem}

\newtheorem{corollary}[theorem]{Corollary}

\newcommand{\Prob}{p}

\begin{document}

\title{Lightweight Detection of a Small Number of Large Errors in a Quantum Circuit}
	\author{Noah Linden}
	\affiliation{School of Mathematics, University of Bristol.  {\tt n.linden@bristol.ac.uk}}
	\author{Ronald de Wolf}
	\affiliation{QuSoft, CWI and University of Amsterdam, the Netherlands.  {\tt rdewolf@cwi.nl}}	
	
	\date{}
	\maketitle

\begin{abstract}
Suppose we want to implement a unitary $U$, for instance a circuit for some quantum algorithm.
Suppose our actual implementation is a unitary $\tilde{U}$, which we can only apply as a black-box. In general it is an exponentially-hard task to decide whether $\tilde{U}$ equals the intended~$U$, or is significantly different in a worst-case norm.
In this paper we consider two special cases where relatively efficient and lightweight procedures exist for this task. 

First, we give an efficient procedure under the assumption that $U$ and $\tilde{U}$ (both of which we can now apply as a black-box) are either equal, or differ significantly in only \emph{one} $k$-qubit gate, where $k=O(1)$ (the $k$ qubits need not be contiguous). Second, we give an even more lightweight procedure under the assumption that $U$ and $\tilde{U}$ are \emph{Clifford} circuits which are either equal, or different in arbitrary ways (the specification of $U$ is now classically given while $\tilde{U}$ can still only be applied as a black-box). Both procedures only need to run $\tilde{U}$ a constant number of times to detect a constant error in a worst-case norm. We note that the Clifford result also follows from earlier work of Flammia and Liu~\cite{flammia&liu:fidelityestimation}
and da Silva, Landon-Cardinal, and Poulin~\cite{SLP:practical}.

In the Clifford case, our error-detection procedure also allows us efficiently to \emph{learn} (and hence correct) $\tilde{U}$ if we have a small list of possible errors that could have happened to $U$; for example if we know that only $O(1)$ of the gates of $\tilde{U}$ are wrong, this list will be polynomially small and we can test each possible erroneous version of $U$ for equality with $\tilde{U}$.
\end{abstract}

\setcounter{page}{1}

\section{Introduction}

With the first tentative steps for implementing quantum computations on larger numbers of qubits (53 qubits in the case of Google's quantum supremacy experiment~\cite{google:supremacy}) comes the need to \emph{verify} whether those implementations actually work as intended. In contrast to classical computations, we cannot just ``open up'' the computer midway through the computation to check whether everything is still on track and then allow the computation to continue, because measurements on the intermediate quantum state typically destroy the  superposition; and learning the quantum state takes exponential effort in the number of qubits in general. Similarly, simulation of  a general $n$-qubit quantum circuit to determine what the intended output should be on a given input state, becomes infeasible if $n\gg 50$.

Even reasonably-well implemented circuits of simple quantum operations (``gates'') can still be marred by many different types of errors: a few large errors (where a gate or qubit is totally wrong or even absent), or many smallish errors (for example slight overrotations), or some combination of both. Strategies are needed to deal with these. In the long run, when we have sufficiently many physical qubits available to encode our logical qubits by error-correcting codes, such faults could in principle all be dealt with by the machinery of fault-tolerant quantum computing. In particular, the ``threshold theorem''~\cite{aharonov&benor:faulttolj} says that arbitrarily long fault-tolerant quantum computing is possible with low overhead, assuming the fault-rate per qubit per time-step is a sufficiently small constant and the errors are not too correlated. But even here, there could be errors due to the mis-specification of the programme to be run.  In the near- to medium-term future we will not have sufficiently many qubits available to do fault-tolerant computing, and we need more ``lightweight'' methods to verify (and hopefully correct) quantum circuits. By lightweight we mean that the verification procedures should not use very complicated quantum operations beyond running $\tilde{U}$ as a black-box, and should need only polynomial (ideally only linear) additional classical effort in the number of qubits and gates of the tested circuits.

We are interested in this paper with testing the full computation, thought of as a black-box, and testing its behaviour on an arbitrary input, not just the all-zeros state (as is important, for example, if the circuit is to be applied as a subroutine within a larger computation). Accordingly, the verification procedure should test for closeness of the ideal circuit~$U$ and the actually implemented circuit~$\tilde{U}$ in a \emph{worst-case norm}.

Let us first discuss what specific norm is appropriate to measure distance between unitaries $U$ and~$\tilde{U}$
(see the survey \cite[Section~5.1]{montanaro&wolf:proptest} and references therein for a more extensive discussion).
When measuring distance between two states, $\ket{\phi}$ and $\ket{\psi}$, the canonical distance measure is the \emph{trace distance}, which is defined as half the difference in Schatten-1 norm between the corresponding density matrices:
\[
D(\ket{\phi},\ket{\psi})=\frac{1}{2}\norm{\ketbra{\phi}{\phi}-\ketbra{\psi}{\psi}}_1.
\]
The trace distance gives exactly the maximal total variation distance difference between the probability distributions obtained from $\ket{\phi}$ and $\ket{\psi}$, respectively, maximized over all possible measurements.
The trace distance between $\ket{\phi}$ and $\ket{\psi}$ turns out to be equal to 
\[
D(\ket{\phi},\ket{\psi})=\sqrt{1-|\braket{\phi}{\psi}|^2}.
\]
This $D(\ket{\phi},\ket{\psi})$ satisfies the triangle inequality, but is not a distance in the strictest sense of the word, because $\ket{\phi}$ and $-\ket{\phi}$ have distance~0 even though they're not equal. This is, however, as it should be, because such global-phase differences have no physical significance.

When comparing different unitaries $U$ and $\tilde{U}$ in the worst case, it is natural to maximize the trace distance between $U\ket{\phi}$ and $\tilde{U}\ket{\phi}$ over all $\ket{\phi}$. This gives the following distance:
\[
D^{max}(U,\tilde{U})=\max_{\ket{\phi}}D(U\ket{\phi},\tilde{U}\ket{\phi})=\max_{\ket{\phi}}\sqrt{1-|\bra{\phi}U^\dagger\tilde{U}\ket{\phi}|^2}.
\]
This is actually the special case of the diamond-norm distance, restricted to the case of unitaries.\footnote{For computing the diamond-norm distance between unitaries, the usual ``stabilization'' by tensoring with identity is not needed; see~\cite[end of Section~5.3]{akn:mixed}, or \cite[Theorem~3.55]{watrous:qit} for a more general statement. This also implies that an implementation~$\tilde{U}$, once verified to have small or no distance from the intended~$U$, can also be applied reliably to a subset of $n$ qubits within a larger computation.} 
Similarly to the trace distance, this distance cannot ``see'' the difference in global phase between $U$ and $e^{i\theta}U$ (unless we can turn the global phase into a relative phase by conditional operations).\footnote{This example also shows why operator-norm difference is not the right worst-case distance measure here: $\norm{U-\tilde{U}}=2$ if $\tilde{U}=-U$, even though $U$ and $-U$ are indistinguishable.}

Detecting the difference between $U$ and $\tilde{U}$ as measured by $D^{max}$ is like finding a needle in a haystack: two $n$-qubit unitaries may have large $D^{max}$-distance  while being equal on all but one of the elements in some particular $2^n$-element basis. The difference would only show up in one out of $2^n$ possible ``directions''. Consider the example where the ideal unitary~$U$ is the $n$-qubit identity and the actual implementation~$\tilde{U}$ is identity with one of the $2^n$ diagonal entries negated; here $D^{max}(U,\tilde{U})$ is large (equal to~1), yet the well-known lower bound for quantum search~\cite{bbbv:str&weak} implies that $\Omega(\sqrt{2^n})$ black-box applications of $\tilde{U}$ are necessary in order to detect the difference from identity with constant probability.
And in a complexity-theoretic context, where the unitaries are not given as a black-box but as explicit polynomial-size quantum circuits, deciding whether $D^{max}(U,\tilde{U})$ is close to~0 or close to~1 is known to be QIP-complete~\cite[Theorem~13]{watrous:qcomplexitysurvey}.
Still, some non-trivial verification can be done in special cases without doing an exponential amount of work, and that is the topic of this paper.

We will consider two types of $U,\tilde{U}$ in the following subsections: (1) arbitrary unitaries, which we can think of as (possibly  very large) circuits over an arbitrary universal set of gates, for instance $\{H,T,\mbox{CNOT}\}$. Here we will be able efficiently to detect large $D^{max}$-distance if $U$ and $\tilde{U}$ differ in only one $k$-qubit gate, with $k=O(1)$. And (2) unitaries corresponding to \emph{Clifford} circuits.
Here we will be able to efficiently detect difference between any two Clifford circuits $U$ and~$\tilde{U}$.
In both cases our procedures only need to run the circuits a \emph{constant} number of times in order to detect a constant distance in worst-case norm.
In case (2), if the number of faulty gates in our Clifford circuit is $O(1)$, then we can actually find what those errors are in polynomial time.

\subsection{Circuits over a universal gate set}

Suppose we want to test whether two $n$-qubit unitaries $U$ and $\tilde{U}$ over an arbitrary gate-set are equal or not. We can apply these unitaries as a black-box, but cannot look inside them. For example, we can think of $U$ as corresponding to an implementation of some $s$-gate quantum circuit on a chip, which for whatever reason we already know to be a correct implementation. $\tilde{U}$ is another chip that has just come off the production line and that is supposed to equal~$U$, but that may or may not be different (faulty) in one or more of the $s$ elementary gates. We want to test whether $U$ and $\tilde{U}$ are either equal, or far in $D^{max}$-distance.

In Section~\ref{ssecusingseparately} we describe a well-known test that compares $U$ and $\tilde{U}$ by effectively comparing their ``Choi states''. By running $U\otimes I$ on the first half of $n$ EPR-pairs, running $\tilde{U}\otimes I$ on the first half of another batch of $n$ EPR-pairs, and comparing the two resulting $2n$-qubit states\footnote{This protocol has been used in various places, and has recently even been implemented on a small quantum computer~\cite{KLPCSC:qcompiling}.
\label{jonasfootnote}Jonas Helsen (personal communication) noted that one can also  do something similar by applying $U$ and $\tilde{U}$ each to their own copy of the same Haar-random $n$-qubit state. This saves half the qubits, but generating two copies of the same Haar-random state (or something pseudo-random like running a 2-design) is not so lightweight. The idea is somewhat similar
to the experimentally motivated method of Elben et al.~\cite{elbenea:verification} to test whether two separate experimental implementations (possibly on very different types of physical hardware platforms) produce approximately the same state by applying the same random product unitary to both and then measuring in the standard basis.}
with a swap-test, we obtain a test with acceptance probability given by
\begin{equation}\label{eq:choitestprob}
\Prob=\frac{1}{2}D(U,\tilde{U})^2,
\end{equation}
where $D$ is an ``average-case'' distance measure defined by:\footnote{\label{note:entfidelity}The quantity $|\frac{1}{2^n}\Tr(U^\dagger\tilde{U})|^2$ is sometimes called the ``entanglement fidelity'' (see~\cite{horodecki:general,nielsen:simpleformula}) between the quantum channels associated to the unitaries $U$ and $\tilde{U}$, and 1 minus the entanglement fidelity (also known as the ``entanglement \emph{in}fidelity'') is the square of our $D(U,\tilde{U})$.
Note that Eq.~(\ref{eq:choitestprob}) also allows to \emph{estimate} this (in)fidelity by repeating the test to estimate~$p$.
This is the key to the results of Flammia and Liu~\cite{flammia&liu:fidelityestimation} discussed in Section~\ref{ssec:flammialiu}.}
\[
D(U,\tilde{U})=\sqrt{1-\left|\frac{1}{2^n}\Tr(U^\dagger\tilde{U})\right|^2}.
\]
The following equality justifies calling $D(U,\tilde{U})$ an ``average case''~\cite[Proposition~21]{montanaro&wolf:proptest}:
\[
D(U,\tilde{U})^2=\frac{2^n+1}{2^n}\int D(U\ket{\phi},\tilde{U}\ket{\phi})^2\,d\phi,
\]
where the integral is according to Haar measure, and $(2^n+1)/2^n$ is very close to~1 already for small~$n$.
Hence our test is sensitive to a difference in trace distance in an \emph{average} direction.\footnote{This formula shows that if $D(U,\tilde{U})$ is lower bounded by a constant, and one chooses $\ket{\phi}$ according to Haar measure (or some 2-design), then with constant probability $D(U\ket{\phi},\tilde{U}\ket{\phi})$ is lower bounded by a (smaller) constant as well. This means we can detect differences between two quantum circuits in $D$-norm by comparing the states $U\ket{\phi}$ and $\tilde{U}\ket{\phi}$ for a randomly chosen $\ket{\phi}$. This observation was used recently in a circuit-verification scheme under the assumption that one can classically simulate and compare the states $U\ket{\phi}$ and $\tilde{U}\ket{\phi}$~\cite{ABKW:stimuli}.}
That is of course much weaker than we want, because $U$ and $\tilde{U}$ can have large distance $D^{max}(U,\tilde{U})$, even when the detection probability of Eq.~(\ref{eq:choitestprob}) is exponentially close to~0. 
However, we show that if $U$ and $\tilde{U}$ differ in only one gate on $k=O(1)$ qubits (in the case where our circuit has some fixed spatial geometry: these qubits need not be contiguous), then the $D^{max}$ and $D$ distances are closely related, and one is large iff the other is large. This gives a relatively lightweight procedure to compare two black-box circuits that differ in at most one $k$-qubit gate. Note that the procedure does not tell us what or where the erroneous gate is.

This really concerns one extreme end of the spectrum of possible ways in which a circuit can fail: the relatively simple situation where \emph{one} $k$-qubit gate is significantly wrong (the $k=O(1)$ qubits need not be contiguous, and the $k$-qubit gate that is wrong could be built up from multiple elementary gates, some of which may be wrong), while the other gates in the circuit are essentially perfect. The picture we have in mind is analogous to a chip, where bits or qubits are led through a physical circuit, on which each gate has its own location. This setting does not really correspond to the current proposals for implementing quantum circuits on superconducting or ion-trap hardware, where typically many of the gates can be slightly faulty, and gradual deterioration is going on all over the place. However, our picture could correspond to optical implementations of quantum computers, where the optical set-up implementing a circuit on fly-by photonic qubits has one erroneous location, while everything else works essentially as intended. It could also correspond to the situation where we have a classical program driving near-perfect quantum hardware, where the \emph{classical} program has one erroneous instruction somewhere, leading to one gate not doing what it's supposed to do (near-perfect quantum hardware that receives the wrong instructions still fails).

As an application, our test can be used to winnow out the faulty circuits from a production line where each circuit has small probability $f$ of having one faulty gate. Using our test we can reduce the fraction of faulty circuits from~$f$ to anything we want (see Section~\ref{ssec:reducingfaultrate}).

What if $\tilde{U}$ has more than one faulty gate compared to~$U$?
One would expect two errors to be no harder to detect than one error. Unfortunately, as we show in Section~\ref{ssec:twofaults}, there are cases where $U$ and $\tilde{U}$ differ significantly in two 1-qubit gates and have large $D^{max}(U,\tilde{U})$-distance, yet the two errors conspire to make $D(U,\tilde{U})$ (and hence the detection probability of our test)  exponentially small.

Our test for arbitrary gate tests assumes the ability to create $2n$ EPR-pairs, to maintain coherence between the two halves of the EPR-pairs during the run of the circuits, and to apply a swap-test to two $2n$-qubit gates. This is reasonably lightweight but not quite as lightweight as we would like our test to be.

\subsection{Clifford circuits}
In order to enable more lightweight testing, we then turn our attention to a specific gate-set. Clifford circuits use the gate-set consisting of the Pauli matrices: 
$$
I=\left(\begin{array}{cc}
1 & 0\\
0 & 1\\
\end{array}\right),~~
X=\left(\begin{array}{cc}
0 & 1\\
1 & 0\\
\end{array}\right),~~
Y=\left(\begin{array}{cc}
0 & -i\\
i & 0\\
\end{array}\right),~~
Z=\left(\begin{array}{cc}
1 & 0\\
0 & 1\\
\end{array}\right),
$$
Hadamard $H$, phase gate $S$, and CNOT.
This gate-set is not universal; it becomes universal when adding for instance a $T$-gate or when we start with certain ``magic states'' as part of our initial state and allow classical conditioning on the outcomes of intermediate one-qubit measurements (using Clifford gates we can then implement a $T$-gate).

We will consider the situation where we would like to implement a Clifford circuit $U$, of which we now have a \emph{classical} description. We also have an implementation of a (possibly different) Clifford circuit $\tilde{U}$ that we can run as a black-box. In Section~\ref{sec:clifford} we give a relatively lightweight procedure for testing (with success probability close to~1) whether $U=\tilde{U}$ or not, which only uses $O(1)$ runs of the black-box circuit $\tilde{U}$ together with single-qubit state preparations at the start, and single-qubit measurements at the end.
In fact, even with \emph{one} run of $\tilde{U}$ we already have probability $\geq 1/4$ of detecting a difference between $U$ and~$\tilde{U}$. This also means that our test still works if the errors are different in each run (i.e., if $\tilde{U}$ is a different erroneous Clifford in different runs).

The reason we can have such a lightweight procedure for testing Clifford circuits, is that such circuits correspond to linear  maps of the set of $n$-qubit Pauli matrices to itself (up to an overall phase $\pm 1$), and that two different such maps actually differ on at least half of the $4^n$ $n$-qubit Paulis. Our test thus selects an $n$-qubit Pauli at random, and indirectly checks (by  appropriate  single-qubit measurements on the state obtained by running $\tilde{U}$ on an appropriate product state) whether $\tilde{U}$ transforms that Pauli as $U$ would have done.
This test is inspired by a test due to Richard Jozsa~\cite{jozsa:cliffordnote}, which however uses $O(n)$ runs of $\tilde{U}$ rather than our $O(1)$ runs.  

In contrast to the procedure of the previous section, 
this test can distinguish any two different Cliffords, and we do not need to make any assumptions about $U$ and $\tilde{U}$ differing in only one gate.
However, if we \emph{do} additionally assume that $U$ and $\tilde{U}$ differ in at most one gate, then we can not only detect the presence of an error, but even find what it is. More generally, if we can compute from $U$ any small list of candidate circuits that is promised to contain $\tilde{U}$, then we can use our test from this section to identify $\tilde{U}$ by running over all $U'$ in our list and testing whether $U'=\tilde{U}$.\footnote{Note that we are not actually learning the specific circuit-implementation for $\tilde{U}$, but only a Clifford circuit for the same unitary~$\tilde{U}$. This is unavoidable because there are many different Clifford circuits implementing the same unitary~$\tilde{U}$, and black-box runs of $\tilde{U}$ cannot see the difference between these different implementations.}  
For example, if we know that the implemented circuit $\tilde{U}$ was obtained from the ideal specification $U$ by $O(1)$ gates that were replaced by other gates, then this list has a size that is only polynomial in the number of qubits and gates of~$U$.  Having found $\tilde{U}$, we have learnt the error(s), which hopefully enables us to correct it (them).

\paragraph{Remark.}
After we finished our section on the above test for equivalence of Clifford circuits, we discovered that this result also follows from earlier work of Flammia and Liu~\cite{flammia&liu:fidelityestimation} and da Silva, Landon-Cardinal, and Poulin~\cite{SLP:practical} about estimating the fidelity between quantum states and between quantum channels, together with the additional observation that distinct Cliffords have noticeably large $D$-distance (equivalently, small entanglement fidelity).
We give the details in Section~\ref{ssec:flammialiu}.

\subsection{Related work}

With the development of medium-size quantum computers, verification of their properties is receiving more and more attention. Here we will mention some of the main approaches and results, referring to the  recent survey \cite{eisertetal:certificationsurvey} and the many references therein for more. 

From a theoretical standpoint, an important recent result is Mahadev's verification protocol~\cite{mahadev:clasveri}. This sits at the end of a long line of works in the area of blind quantum computation~\cite{AS:blindqc,BFK:blindqc} where a single verifier (who should be as efficient and as classical as possible) checks a quantum computation by interacting with one or more polynomial-time quantum provers. 
Mahadev's protocol allows a purely-classical polynomial-time machine to verify the computation of a polynomial-time quantum machine (under reasonable cryptographic assumptions). However, even though everything in Mahadev's protocol is polynomial, and hence ``efficient'' from a theoretical perspective, in practice the protocol is anything but lightweight: it leads to very significant overheads on the side of the quantum computation, and several rounds of communication between the quantum computer and the classical verifier. It is also designed to test the computation starting from a fixed initial state, so does not test according to a worst-case distance measure. Mahadev's 4-message protocol was subsequently improved to a non-interactive and zero-knowledge protocol in~\cite{ACGH:verification}, but that improved protocol is not very lightweight either. 

A much more bottom-up approach to verification is to test the building blocks of the quantum algorithm: the elementary gates. There have been some positive results on testing universal sets of quantum gates, for instance~\cite{dmms:selftestj}.
However, testing gates in isolation is not enough to verify their behavior in the context of a larger circuit, where the surrounding components may adversely affect gates that would have worked fine in isolation. Randomized benchmarking~\cite{EAZ:scalable,DCEL:2designs,PRYSB:randomized} is an approach to test sequences of gates: roughly speaking one runs a random sequence of gates from a fixed gate set (often restricted to Clifford circuits on a small number of qubits) followed by their inverse, and then tests to what extent the resulting operation is the identity, as it should be. This approach beautifully isolates the average entanglement fidelity of the gates (see footnote~\ref{note:entfidelity}), from state preparation and measurement (``SPAM'') errors.
Note that the average entanglement fidelity of the gates, which is what randomized benchmarking tries to measure, is an average-case measure that may or may not give information about the worst-case errors of these gates~\cite{wallman:errorrates,KLDF:comparing}.

Closer to the second part of this paper is the work of Low~\cite{low:clifford}, who studied efficient testing and even identifying (learning) of Clifford circuits. He showed how to fully learn an unknown Clifford circuit $U$ using $O(n)$ runs of $U$ and $U^\dagger$, but assuming the ability to run $U^\dagger$ is a  stronger assumption than we are willing to entertain here.  
Low also points (at the end of his Section III.B) to work of Harrow and Winter~\cite{harrow&winter:howmany} which implies that $O(n^2)$ runs of $U$ suffice to learn it (without using $U^\dagger$), but their work is information-theoretic in nature and assumes the ability to do complicated joint measurements on the $O(n^2)$ output-states of the runs of~$U$.
The general philosophy we espouse here (looking for lightweight schemes) is also embodied in the verification protocol described by Jozsa and Strelchuk in~\cite{jozsa&strelchuk:verification}.
Last but not least, we already mentioned the very related work of Flammia and Liu~\cite{flammia&liu:fidelityestimation}
 and~\cite{SLP:practical}, which we discuss in Section~\ref{ssec:flammialiu}.

\section{Testing circuits over an arbitrary gate set}\label{sec:generalgates}

\subsection{Using the two circuits separately}\label{ssecusingseparately}

In this section we study the situation where we have two $s$-gate quantum circuits, $U$ and $\tilde{U}$, over an arbitrary set of one- and two-qubit gates. We can run these in a black-box fashion and want to test whether they are either equal, or substantially different in operator norm. We will give a relatively lightweight test that works if $U$ and $\tilde{U}$ differ in at most one gate.

We start by reminding the reader of a simple test that is sensitive to average-case distance between $U$ and $\tilde{U}$ (see~\cite[Section~5.1.3]{montanaro&wolf:proptest} and references therein).
We will assume it is possible to create a maximally entangled state on $2n$ qubits; a simple circuit that starts from $\ket{0^{2n}}$ and applies $n$ Hadamard gates and $n$ CNOTs will do this. We also assume we can do controlled-swap gates. Such 3-qubit gates are not quite as lightweight as we'd ideally like to be, but still much lighter than universal quantum computation.

Now consider the following test:
\begin{enumerate}
\item Run $U\otimes I$ on a $2n$-qubit maximally entangled state to produce state $\ket{\psi_U}$.
\item Run $\tilde{U}\otimes I$ on another $2n$-qubit maximally entangled state to produce state $\ket{\psi_{\tilde{U}}}$.
\item Run a swap-test on $\ket{\psi_U}$ and $\ket{\psi_{\tilde{U}}}$ and output the measured bit.\footnote{The swap-test~\cite{bcww:fp} starts with an auxiliary qubit in the $\ket{+}$ state, swaps the two registers conditioned on the auxiliary qubit, and then measures the auxiliary qubit in the Hadamard basis.}
\end{enumerate}
This test uses $O(n+s)$ gates.
It is easy to calculate the probability that the test outputs~1:
\[
\Prob=\frac{1}{2}-\frac{1}{2}\left|\braket{\psi_U}{\psi_{\tilde{U}}}\right|^2,
\]
and that
\[
\braket{\psi_U}{\psi_{\tilde{U}}} =  \frac{1}{2^n}\Tr(U^\dagger \tilde{U}).
\]
This gives a relation between $\Prob$ and the average-case distance $D(U,
\tilde{U})=\sqrt{1-|\frac{1}{2^n}\Tr(U^\dagger \tilde{U})|^2}$ defined in the introduction:
\[
\Prob=\frac{1}{2}D(U,\tilde{U})^2.
\]
If $U$ and $\tilde{U}$ are equal (up to global phase) then $p$ will be 0, and otherwise $p$ will be positive.
Measurement outcome~1 thus tells us that $U$ and $\tilde{U}$ are different (by more than a global phase). The detection probability is large iff $D(U,\tilde{U})$ is large.  
This test will therefore be useful, for example, if $\tilde{U}$ were a version of $U$ hit by random errors, because random errors tend to create deviations in many ``directions'' simultaneously and hence give a non-negligible distance~$D(U,\tilde{U})$.
However, our main focus here is to design a test that is sensitive to the \emph{worst-case} distance $D^{max}(U,\tilde{U})$, because if that distance is small, then $U$ and $\tilde{U}$ produce approximately the same states no matter what initial state they are applied to.
In general the relation between these worst-case and average-case norms is fairly weak. For example, if $U=I$ and $\tilde{U}$ has one of its diagonal entries set to $-1$, then $D^{max}(U,\tilde{U})=1$ but $D(U,\tilde{U})=\sqrt{4/2^n-4/2^{2n}}$ is exponentially small. The above test will thus have exponentially small probability of detecting the large $D^{max}$-distance in this case.

We now show that at least the gap cannot be much more than in the previous example.

\begin{theorem}\label{th:DvsDmax}
If $U$ and $\tilde{U}$ are $n$-qubit unitaries, then 
$D^{max}(U,\tilde{U})\leq 2^{(n+1)/2} D(U,\tilde{U})$.
\end{theorem}

\begin{proof}
Let $\mu=\min_{\ket{\phi}}|\bra{\phi}U^\dagger\tilde{U}\ket{\phi}|$, and $\ket{\phi}$ be a minimizing state.
Let $\cal B$ be an orthonormal basis that contains $\ket{\phi}$ as one of its $2^n$ states.
We have 
\[
\left|\Tr(U^\dagger\tilde{U})\right|  =\left|\sum_{b \in {\cal B}} \bra{b}U^\dagger \tilde{U}\ket{b}\right|
 \leq \sum_{b \in{\cal B}} \left| \bra{b}U^\dagger\tilde{U}\ket{b}\right|
\leq 2^n-1+|\bra{\phi}U^\dagger\tilde{U}\ket{\phi}| 
   = 2^n - 1 + \mu.
\]
We now bound
\begin{align*}
D(U,\tilde{U})^2 &=   1 -  \left|\frac{1}{2^n}\Tr(U^\dagger\tilde{U})\right|^2 \geq 1- (1-(1-\mu)/2^n)^2= \frac{1-\mu}{2^n}\left(2-\frac{1}{2^n}+\frac{\mu}{2^n}\right)\\
& \geq \frac{1-\mu}{2^n}\cdot 1 \geq \frac{1-\mu}{2^n}\cdot\frac{1+\mu}{2} = \frac{1}{2^{n+1}}(1-\mu^2) = \frac{1}{2^{n+1}}D^{max}(U,\tilde{U})^2,
\end{align*}
which implies the inequality of the theorem.
\end{proof}

The above theorem is a strengthening (for the special case of unitaries) of a more general but quadratically weaker bound relating these two distances due to Magesan, Gambetta, and Emerson~\cite{MGE:randbench}, and used for instance in~\cite[Eq.~(3)]{wallman:errorrates} and \cite[p.~2]{KLDF:comparing}.

Now we make the simple but powerful observation that if $U$ and $\tilde{U}$ differ only in one $k$-qubit gate ($G$ vs $\tilde{G}$), then the two norms are within a factor of roughly~$2^{k/2}$ of one another. Specifically, let $U=U_1(G\otimes I_{2^{n-k}})U_2$ and $\tilde{U}=U_1(\tilde{G}\otimes I_{2^{n-2}})U_2$, where $U_1$ and $U_2$ are arbitrary unitaries, and $G$ and $\tilde{G}$ are $k$-qubit gates. For notational simplicity we wrote $G$ and $\tilde{G}$ as acting on the first $k$ qubits of the state, but in fact they may act on any subset of $k$ of the $n$ qubits, not necessarily contiguous. 
We have
\[
\frac{1}{2^n}\Tr(U^\dagger \tilde{U})
=\frac{1}{2^n}\Tr(U_2^\dagger(G^\dagger\otimes I_{2^{n-k}})U_1^\dagger\cdot U_1(\tilde{G}\otimes I_{2^{n-k}})U_2)
=\frac{1}{2^n}\Tr(G^\dagger\tilde{G}\otimes I_{2^{n-k}})=\frac{1}{2^k}\Tr(G^\dagger\tilde{G})
\]
and hence $D(U,\tilde{U})=D(G,\tilde{G})$.
We also have 
\[
\min_{\ket{\phi}}\left|\bra{\phi}U^\dagger\tilde{U}\ket{\phi}\right|
=\min_{\ket{\phi}}\left|\bra{\phi}(G^\dagger\otimes I_{2^{n-k}})(\tilde{G}\otimes I_{2^{n-k}})\ket{\phi}\right|
=\min_{\ket{\psi}}\left|\bra{\psi}G^\dagger\tilde{G}\ket{\psi}\right|
\]
and hence $D^{max}(U,\tilde{U})=D^{max}(G,\tilde{G})$. 

Therefore, using Theorem~\ref{th:DvsDmax}, the probability of detecting a difference between $U$ and $\tilde{U}$ is
\[
\Prob=\frac{1}{2}D(U,\tilde{U})^2=\frac{1}{2}D(G,\tilde{G})^2\geq
\frac{1}{2^{k+2}}D^{max}(G,\tilde{G})^2
=\frac{1}{2^{k+2}}D^{max}(U,\tilde{U})^2.
\]
In particular, if the worst-case distance is $D^{max}(U,\tilde{U})\geq\eps$ and $k=O(1)$ (say, $U$ and $\tilde{U}$ differ only in one $k$-qubit gate, or in one block of errors that affects only $k$ qubits, not necessarily contiguous), then our detection probability is $\Prob=\Omega(\eps^2)$.
We can efficiently increase this detection probability to close to~1: if we run $O(\log(1/\delta)/\eps^2)$  tests, then if $U$ and $\tilde{U}$ are equal then all tests will output~0, while if $D^{max}(U,\tilde{U})\geq\eps$ then with probability $\geq 1-\delta$ at least one of the tests will output~1.

The $1/\eps^2$-factor in the number of tests could be improved to $1/\eps$ using amplitude amplification~\cite{bhmt:countingj}, but that would be a much less lightweight procedure: it also requires the ability to apply controlled versions of $U$ and $\tilde{U}$ as well as of their inverses, which may be technologically rather demanding. In any case, if we can apply inverses then there is an easier test that only uses $n$ EPR-pairs instead of $2n$: apply $\tilde{U}$ and $U^{-1}$ to the first half of a $2n$-qubit maximally entangled state, reverse the $H$s and CNOTs that prepared the entangled state, and check (by a measurement in the computational basis) whether you get back $\ket{0^{2n}}$, as this provides an estimate of $D(U,\tilde{U})$.

\subsection{If we can apply the circuits conditionally}

In the case where we cannot apply conditional versions of $U$ and $\tilde{U}$, like in the previous section, differences in their global phases are physically meaningless and we cannot detect them.
Now suppose we have slightly more power: we \emph{can} apply $U$ and $\tilde{U}$ in a conditional manner, but not their inverses. This allows for a slightly more efficient test that uses $2n+1$ qubits instead of $4n$:
\begin{enumerate}
\item Prepare $H\ket{0}$ tensored with a $2n$-qubit maximally entangled state ($2n+1$ qubits in total).
\item Conditioned on the first qubit being $\ket{0}$, apply $U$ to the first $n$-qubit block;\\
conditioned on the first qubit being $\ket{1}$, apply $\tilde{U}$ to the first $n$-qubit block.
\item Apply $H$ to the first qubit and measure it.
\end{enumerate}
The probability that the above algorithm outputs~1 is
\[
\Prob=\frac{1}{2}-\frac{1}{2}\mathbb{R}(\braket{\psi_U}{\psi_{\tilde{U}}})=\frac{1}{2}-\frac{1}{2\cdot 2^n}\mathbb{R}(\Tr(U^\dagger \tilde{U})).
\]
Note that $\Tr(U^\dagger \tilde{U})$ is not squared here, in contrast to the expression for the probability in the previous section. 
Hence this test is sensitive to the relative phase between $U$ and $\tilde{U}$.
In particular, if $\tilde{U}=U$ then $p=0$, while if $\tilde{U}=-U$ then $\Prob=1$.

By similar calculations as before, if the  only difference between $U$ and $\tilde{U}$ is in one $k$-qubit gate ($G$ vs $\tilde{G}$), then we have
\[
\frac{1}{2^n}\Tr(U^\dagger\tilde{U})=\frac{1}{2^k}\Tr(G^\dagger\tilde{G})
\]
and
\[
\Prob=\frac{1}{2}-\frac{1}{2\cdot 2^k}\mathbb{R}(\Tr(G^\dagger \tilde{G})).
\]

\subsection{Reducing the fault-rate in a production line of circuits}\label{ssec:reducingfaultrate}

Suppose we have a production line that is intended to produce identical circuits that implement a particular unitary $U$. Like everything else in life, the production line is not perfect. Assume that each circuit is perfect (i.e., equal to $U$) with probability $1-f$ and faulty with probability $f$, meaning its $D^{max}$-distance from the ideal $U$ is at least $\eps$; for example because $U$ and $\tilde{U}$ differ in exactly one gate like before.\footnote{For simplicity we will ignore the case of positive but smaller error $<\eps$.} If we don't do anything, we expect a fraction of roughly $f$ of the circuits to be faulty. We would like to reduce this fraction by efficiently identifying the faulty circuits. We can achieve this by comparing the circuits against each other, using the fact that most are probably correct. Note that we are not assuming here that we can run the ideal $U$ as a black-box.

Assume we have a test that, given two circuits $U_1$ and $U_2$, can distinguish between the cases $U_1=U_2$ (up to global phase) and $D^{max}(U_1,U_2)\geq\eps$ with success probability $\geq 2/3$
(for example, our test from Section~\ref{ssecusingseparately} will do that if the distance is due to one faulty gate).
Note that we can reduce the error probability of this test from $1/3$ to small $\delta$ by running it $O(\log(1/\delta))$ times and taking the majority outcome among those runs.

Let us take a batch of $n$ circuits coming off the production line, with $n$ odd. By a Chernoff bound, the probability that more than half of them are faulty is at most $e^{-D(1/2||f)n}$, where $D(p||q)=p\ln(p/q)+(1-p)\ln((1-p)/(1-q))$ is the Kullback-Leibler divergence (a.k.a.\ relative entropy) between binary distributions with probabilities $p$ and $q$ respectively, measured in nats rather than bits. If $f$ is bounded away from $1/2$, then $D(1/2||f)=\Omega(1)$ and $e^{-D(1/2||f)n}$ is exponentially small in $n$.
Now suppose we run our test on each of the $\binom{n}{2}$ pairs in the batch, with error probability reduced to $\delta\ll 1/n^2$. Then, except with probability $p_E\leq\binom{n}{2}\delta+e^{-D(1/2||f)n}\ll 1$, all tests succeed \emph{and} more than half of the circuits in the batch are correct.
Condition on this event below.

Each circuit in the batch will be involved in $n-1$ tests. For every good circuit, at least half of the tests it is involved in will be with other good circuits and hence will say ``equal''. For faulty circuits, more than half of the tests it is involved in will be with good circuits and hence these will say ``not equal''. Accordingly, if we throw away the circuits where more than half of the tests say ``not equal'', then we will exactly eliminate the faulty circuits from this batch. 

With probability $p_E$, the event we conditioned on did not happen, but the worst that can occur in that case is that we err on all $n$ circuits in that batch, in the sense of throwing away all good circuits from the batch and keeping all faulty ones. Since $p_E$ is exponentially small in~$n$, this bad event only negligibly affects the expected fraction of circuits we mishandled.

By choosing the batch-size $n$ large enough, we can thus reduce the expected fault rate from $f$ to anything we want. The number of black-box runs used for analyzing each batch of $n$ circuits, is $O(\binom{n}{2}\log(1/\delta))=O(n^2\log(n))$.

\subsection{Detecting \emph{two} faults is hard for our test in the worst case}\label{ssec:twofaults}

The test of Section~\ref{ssecusingseparately} works to detect a one-gate error, because if only one gate is affected then there is a fairly tight relation between average-case distance $D(U,\tilde{U})$ that we can test for, and the worst-case distance $D^{max}(U,\tilde{U})$ that we would like to test for.  What if there are \emph{two} faulty gates in $\tilde{U}$?  One might expect that detecting two errors should be easier than detecting one, but unfortunately this turns out to be false (at least in the worse case) because the two faults can conspire to destroy the close relation between the worst-case and average-case distance measures.

Here's a simple example. Let $V$ be the $n$-qubit C$^{n-1}$NOT gate, which applies an $X$-gate to the last qubit conditioned on the first $n-1$ qubits being in basis state $\ket{1^{n-1}}$. Suppose $U=(I\otimes H)V(I\otimes H)$ and $\tilde{U}=V$. In other words, the intended $H$-gates on the last qubit at the start and the end of the circuit are replaced by identities, so only two of the gates of $U$ are faulty. Because $HXH=Z$, we have
\[
U=\left(
\begin{array}{cccc}
1 & & & \\
& \ddots & &\\
& & 1 & \\
& & & Z
\end{array}
\right)
\mbox{ and }
\tilde{U}=\left(
\begin{array}{cccc}
1 & & & \\
& \ddots & &\\
& & 1 & \\
& & & X
\end{array}
\right).
\]
The matrix $U^\dagger \tilde{U}$ has $ZX=iY$ in its lower-right corner. Hence $\min_{\ket{\phi}}|\bra{\phi}U^\dagger\tilde{U}\ket{\phi}|=0$, as witnessed for instance by taking $\ket{\phi}=\ket{1^n}$.
This implies $D^{max}(U,\tilde{U})=1$. 
On the other hand, $\Tr(U^\dagger \tilde{U})=2^n-2$, hence $D(U,\tilde{U})^2=1-(1-2/2^n)^2\approx 4/2^n$.
The latter implies that one run of our test only has exponentially small probability of detecting the large $D^{max}$-distance between $U$ and $\tilde{U}$. In other words, our test fails miserably to detect two or more adversarially placed faulty gates.

\section{Testing Clifford circuits}\label{sec:clifford}

Let ${\cal P}=\{I,X,Y,Z\}$ be the set of 1-qubit Paulis.
Note that non-identity Paulis anti-commute ($XZ=-ZX$ etc.)
 and that $Y=iXZ$.
Let ${\cal P}^n=\{I,X,Y,Z\}^{\otimes n}$ be the set of $4^n$ $n$-qubit Paulis.
These matrices are unitary and Hermitian, and hence self-inverse.

An $n$-qubit \emph{Clifford circuit} $U$ consists of Pauli gates, Hadamard gates ($H$), phase gates ($S$),
and CNOT gates.
These are exactly the unitaries that map (by conjugation) all elements of ${\cal P}^n$ to elements of ${\cal P}^n$, possibly with an overall phase of $\pm 1$.
We assume there are no intermediate measurements of qubits in the middle of the circuit; these may all be pushed to the end using some auxiliary qubits and CNOTs. 

In this section we will deal with the situation where we want to implement an $n$-qubit Clifford circuit $U$, which we know fully (i.e., we have a classical description of it). Instead we have a Clifford circuit $\tilde{U}$ that we can apply as a black-box. Our goal here is to test whether $U=\tilde{U}$ and, if not, to figure out how they differ so we can correct the errors.

\subsection{What it means for two Clifford circuits to be different}

As mentioned, conjugation by a Clifford circuit $U$ maps elements of ${\cal P}^n$ to elements of ${\cal P}^n$, up to an overall phase $\pm 1$, and it is well known that this map (ignoring the $\pm 1$s) corresponds to a linear map $\mathbb{F}_2^{2n}\to\mathbb{F}_2^{2n}$, where $\mathbb{F}_2 $ is the field of two elements.
Here we represent $I$ by $00\in\mathbb{F}_2^2$, $X$ by 10, $Z$ by 01, and $Y$ by 11, so we may identify an $n$-qubit Pauli with an element of $\mathbb{F}_2^{2n}$. For example, we can identify $P=X\otimes Z$ with the 4-bit vector $(1,0,0,1)^T$.
The correspondence between a Clifford and its associated linear map with signs seems to be folklore. It can be derived from the connection with the symplectic group, see for instance~\cite[Section~I.A]{koenig&smolin:Clifford} (see also~\cite[Section~II.B]{gross:hudson}, though that applies to qudits of odd dimension).
We give a simple proof below for completeness.

\begin{theorem}[folklore]\label{th:cliffordlinear}
Let $U$ be an $n$-qubit Clifford circuit, and define the associated map $U:{\cal P}^n\to\pm{\cal P}^n$ by $U(P)=U P U^\dagger$. There exists an invertible matrix $M_U\in\mathbb{F}_2^{2n\times 2n}$ such that $U(P)\in\{M_UP,-M_UP\}$ (where with slight abuse of notation we view $P$ both as an $n$-qubit Pauli and as an element of $\mathbb{F}_2^{2n}$).
\end{theorem}

\begin{proof}
The circuit $U$ is just a composition of Pauli gates, $H$, $S$, and CNOT gates. Hence it suffices to prove the theorem for each of these gates and then compose the linear maps.

First, when conjugating a 1-qubit Pauli $P$  with a 1-qubit Pauli gate $U$, we just get $P$ back, with a minus sign if $P$ and $U$ anti-commute; we ignore the sign for the purposes of this theorem. The corresponding matrix $M_U$ is just the identity.

Second, conjugation by $H$ maps 1-qubit Paulis to 1-qubit Paulis as:
\begin{quote}
$I\to I$, $X\to Z$, $Z\to X$, $Y\to -Y$
\end{quote}
In the 2-bit representation (ignoring the $\pm 1$) this corresponds to $2\times 2$ matrix
$
\left(\begin{array}{cc}
0 & 1\\
1 & 0\\
\end{array}\right)$  over $\mathbb{F}_2$.
 
Third, conjugation by $S$ maps
\begin{quote}
$I\to I$, $X\to Y$, $Z\to Z$, $Y\to -X$
\end{quote}
In the 2-bit representation (ignoring the $\pm 1$) this corresponds to the $2\times 2$ matrix
$
\left(\begin{array}{cc}
1 & 0\\
1 & 1\\
\end{array}\right)$.

Fourth, conjugation by CNOT maps 2-qubit Paulis to 2-qubit Paulis as given for instance in Figure~3 of~\cite{kruw:noisethresholdj}. It may be verified that in the 4-bit representation this map corresponds to the following $4\times 4$ matrix:
\[
\left(\begin{array}{cccc}
1 & 0 & 0 & 0\\
0 & 1 & 0 & 1\\
1 & 0 & 1 & 0\\
0 & 0 & 0 & 1\\
\end{array}\right).
\]

\vspace*{-2em}
\end{proof}

Clearly, if $M_U$ and $M_{\tilde{U}}$ are different matrices, then $U$ and $\tilde{U}$ must be different Clifford circuits.
However, different Clifford circuits can induce the same matrix $M_U$.
A simple example is a circuit~$U$ that only consists of Pauli gates: if we ignore the $\pm 1$, then conjugation by $U$ is simply the identity map on ${\cal P}^n$, so all Pauli circuits induce the same $M_U=I$.
We now show that this basically describes the only case where different Cliffords induce the same $M_U$:

\begin{corollary} 
Suppose $n$-qubit Clifford circuits $U$ and $\tilde{U}$ have the same induced matrix $M_U$ in Theorem~\ref{th:cliffordlinear}. Then there exists an $R\in{\cal P}^n$ such that conjugation by $U$ and conjugation by $R\tilde{U}$ are the same map on the set of all density matrices.
\end{corollary}

\begin{proof}
First, by right-multiplying $U$ and $\tilde{U}$ with $\tilde{U}^\dagger$, we may assume without loss of generality that $\tilde{U}=I$ and hence $M_{U}=M_{\tilde{U}}=I$. We now want to show that $U$ corresponds to some $R\in{\cal P}$.

Since $M_U=I$, conjugation by $U$ maps each $P\in{\cal P}^n$ to itself, times a sign $s_P$. Since every density matrix $\rho$ is a linear combination of $P\in{\cal P}$, these signs fully determine the action of $U$ on all density matrices: if $\rho=\sum_P a_PP$, then $U\rho U^\dagger=\sum_P a_P UPU^\dagger=\sum_P a_P s_P P$. 

Let us first consider the $n$ signs $s_{X_j}$ induced by the action of $U$ on  $X_j=I^{\otimes j-1}\otimes X\otimes I^{\otimes n-j}$ (for $j=1,\ldots,n$),
and the $n$ signs $s_{Z_j}$ corresponding to $Z_j=I^{\otimes j-1}\otimes Z\otimes I^{\otimes n-j}$.
We now show that we can choose a (unique) $R\in{\cal P}^n$ consistent with all the signs $s_{X_j}$ and $s_{Z_j}$. Consider $j=1$.
\begin{center}
\begin{tabular}{l}
If $s_{X_1}s_{Z_1}=++$, then we choose $R_1=I$ (because $IXI=+X$ and $IZI=+Z$).\\ 
If $s_{X_1}s_{Z_1}=+-$, then we choose $R_1=X$ (because $XXX=+X$ and $XZX=-Z$).\\ 
If $s_{X_1}s_{Z_1}=-+$, then we choose $R_1=Z$  (because $ZXZ=-X$ and $ZZZ=Z$).\\ 
If $s_{X_1}s_{Z_1}=--$, then we choose $R_1=Y$ (because $YXY=-X$ and $YZY=-Z$). 
\end{tabular}
\end{center}
Similarly choose $R_2,\ldots,R_n\in\{I,X,Y,Z\}$ consistent with the signs $s_{X_2}s_{Z_2},\ldots,s_{X_n}s_{Z_n}$.

We now claim that this choice of $R$ (which has $M_R=I$, like all Pauli circuits) not only has the same signs $s_P$ as $U$ for all $P\in\{X_1,\ldots,X_n,Z_1,\ldots,Z_n\}$,
but in fact  has the same signs $s_P$ for \emph{all} $4^n$ $P\in{\cal P}^n$.
To that end, fix an arbitrary $P$, and write it as 
$$
P=cX_1^{a_1}Z_1^{b_1}\cdots X_n^{a_n}Z_n^{b_n},
$$ 
for some $a_1,\ldots,a_n,b_1,\ldots,b_n\in\01$, and some overall phase $c\in\{1,-1,i,-i\}$ which comes from the fact that $Y=iXZ$. Inserting $I=U^\dagger U$ in many places, we can write
\begin{align*}
s_P P & = UPU^\dagger\\
& =U(cX_1^{a_1}Z_1^{b_1}\cdots X_n^{a_n}Z_n^{b_n})U^\dagger\\
& = cUX_1^{a_1}U^\dagger U Z_1^{b_1}U^\dagger U\cdots U^\dagger U X_n^{a_n}U^\dagger U Z_n^{b_n}U^\dagger\\
& = c(UX_1^{a_1}U^\dagger)(U Z_1^{b_1}U^\dagger)\cdots (U X_n^{a_n}U^\dagger)(U Z_n^{b_n}U^\dagger)\\
& =
c(s^{a_1}_{X_1}X_1^{a_1})(s^{b_1}_{Z_1}Z_1^{b_1})\cdots (s^{a_n}_{X_n}X_n^{a_n})(s^{b_n}_{Z_n}Z_n^{b_n})\\
& = (\prod_{j=1}^n s^{a_j}_{X_j}s^{b_j}_{Z_j})(cX_1^{a_1}Z_1^{b_1}\cdots X_n^{a_n}Z_n^{b_n})\\
& = \prod_{j=1}^n s^{a_j}_{X_j}s^{b_j}_{Z_j}\, P.
\end{align*}
This shows that $s_P=\prod_{j=1}^n s^{a_j}_{X_j}s^{b_j}_{Z_j}$, so all $4^n$ signs $s_P$
are fully determined by the $2n$ signs $s_{X_1},s_{Z_1},\ldots,s_{X_n},s_{Z_n}$.
But by the same calculation, $R$ induces exactly the same signs for all $P\in{\cal P}^n$.
Hence conjugation by $U$ and $R$ are the same map on ${\cal P}^n$ (and by linearity  are the same map on all $n$-qubit density matrices).
\end{proof}

\subsection{Our test for detecting a difference between two Clifford circuits}\label{ssec:cliffordtest}

The previous theorems can be used to design an efficient test to detect whether two Clifford circuits (one given classically, the other as a quantum black-box) are equal or not. The test is based on the observation (used for instance in Freivalds's well-known randomized algorithm for verifying matrix multiplication~\cite{freivalds:matrixmult}) that one can detect whether two matrices are equal by comparing their images on a random vector: if the matrices are equal then these images will be the same, but if the two matrices are different then these images will be different with high probability.
In our scenario, if two Clifford circuits $U$ and $\tilde{U}$  are different by more than an $n$-qubit Pauli, then the associated maps $U:{\cal P}^n\to\pm{\cal P}^n$ and $\tilde{U}:{\cal P}^n\to\pm{\cal P}^n$ will give different $n$-qubit Paulis (even when ignoring their signs) on at least half of all $4^n$ Paulis:

\begin{theorem}\label{th:halfthe paulis}
Let $U$ and $\tilde{U}$ be $n$-qubit Clifford circuits that have distinct associated matrices $M_U$ and $M_{\tilde{U}}$ (equivalently, conjugation by $U$ and $R\tilde{U}$ are distinct maps for all $R\in{\cal P}^n$).
Then for at least $\frac{1}{2}4^n$ of the $P\in{\cal P}^n$, $M_U P\neq M_{\tilde{U}}P$ .
\end{theorem}

\begin{proof}
Consider the matrix $M_U-M_{\tilde{U}}\in\mathbb{F}_2^{2n\times 2n}$. This is a nonzero matrix, hence its kernel has dimension at most $2n-1$, which means that $(M_U-M_{\tilde{U}})P=0$ for at most $2^{2n-1}$ different $P$s. Therefore $M_U P\neq M_{\tilde{U}}P$ for at least $2^{2n}-2^{2n-1}=\frac{1}{2}4^n$ of the $P\in{\cal P}^n$.
\end{proof}

Of course, it is possible that $U$ and $\tilde{U}$ only differ by an $n$-qubit Pauli, and we have to consider that case separately.

Now suppose we have a Clifford circuit $\tilde{U}$ that is intended to implement a known Clifford circuit~$U$. We can run $\tilde{U}$ but not its inverse, and want to test whether it  indeed equals the intended~$U$.
Our test starts by choosing a uniformly random $P\in{\cal P}^n$. We compute $U^\dagger(P)=U^\dagger P U$,\footnote{A classical computer can do this in time linear in the number of gates of~$U$: use the $2n$-bit representation and update this gate-by-gate according to the action of the Clifford gates as described in the proof of Theorem~\ref{th:cliffordlinear}; also keep track of the overall phase $\pm 1$. Note that we want to do this for $U^\dagger$ so we have to reverse the order of gates given by~$U$, and invert the gates (which only affects the $S$-gate, since the other Clifford gates are self-inverse).}
which is a signed $n$-qubit Pauli $Q=sQ_1\otimes\cdots\otimes Q_n\in\pm{\cal P}^n$. 
Note that if we start with an eigenstate of~$Q$ and apply $U$ to it, then we obtain an eigenstate of~$P$ itself, with the same eigenvalue.
Our test prepares a tensor-product eigenstate $\ket{\psi_{in}}$ of~$Q$ as follows: 
\begin{center}
\begin{tabular}{l}
for $j=1,\ldots,n$:\\
\hspace*{1em} if $Q_j\in\{X,Y,Z\}$, then set the $j$th qubit of $\ket{\psi_{in}}$ to either the $+1$-eigenstate or the\\ \hspace*{2em}$-1$-eigenstate of $Q_j$, each with probability 1/2;\\ 
\hspace*{1em} if $Q_j=I$, then set the $j$th qubit of $\ket{\psi_{in}}$ to the $+1$-eigenstates $\ket{0}$ or $\ket{1}$, each\\ 
\hspace*{2em}with probability $1/2$ (equivalently, we can think of this as the maximally mixed\\
\hspace*{2em}state, $\frac{1}{2}\ketbra{0}{0}+\frac{1}{2}\ketbra{1}{1}$). 
\end{tabular}
\end{center}
By construction $\ket{\psi_{in}}$ is an eigenstate of $Q$, with an eigenvalue $\lambda\in\{+1,-1\}$ that we know. 
Now we run $\tilde{U}$ on $\ket{\psi_{in}}$ and measure  the $\pm 1$-valued observable $P$ on state $\tilde{U}\ket{\psi_{in}}$.

If $U=\tilde{U}$, then the measurement gives the known value $\lambda$ as outcome, with probability~1.
However, we claim that if $U$ and $\tilde{U}$ are different Cliffords, then we will see the opposite outcome $-\lambda$ with probability at least~$1/4$.
To prove that claim we make a case-distinction for the two ways in which $U$ and $\tilde{U}$ can differ (our test doesn't need to know which of the two cases applies).

\medskip

\noindent
{\bf Case 1:} The matrices $M_U$ and $M_{\tilde{U}}$ are distinct.\\
Let $\tilde{Q}=\tilde{U}^\dagger(P)=\tilde{s}\tilde{Q}_1\otimes\cdots\otimes\tilde{Q}_n\in\pm{\cal P}^n$. We don't know what $\tilde{Q}$ is since we don't know what $\tilde{U}$ is. However, by Theorem~\ref{th:halfthe paulis} we have $Q_1\otimes\cdots\otimes Q_n\neq \tilde{Q}_1\otimes\cdots\otimes\tilde{Q}_n$ with probability at least 1/2, over our random choice of~$P$. 
In this case, measuring $P$ on $\tilde{U}\ket{\psi_{in}}$ will give a value different from $\lambda$ with probability~$1/2$, which can be seen as follows, by examining the different ways in which $Q$ and $\tilde{Q}$ could differ (ignoring their overall signs, which do not affect the probabilistic argument below):
\begin{enumerate}
\item There is a location $j$ where $Q_j,\tilde{Q}_j\in\{X,Y,Z\}$ but $Q_j\neq \tilde{Q}_j$.
$\ket{\psi_{in}}_j$ is a $\pm 1$-eigenstate of $Q_j$ but not of $\tilde{Q}_j$. It is a property of the eigenstates of the non-identity Paulis that $\bra{\psi_{in}}_j \tilde{Q}_j \ket{\psi_{in}}_j=0$, which means that the $j$th qubit will contribute a uniformly random sign to the measurement outcome.
\item There is a location $j$ where $Q_j=I$ and $\tilde{Q}_j\in\{X,Y,Z\}$. Then the $j$th qubit has been set to the maximally mixed state, which is an equal mixture of the $+1$-eigenstate and the $-1$-eigenstate of $\tilde{Q}_j$.
Again, the $j$th qubit will contribute a uniformly random sign to the measurement outcome.
\item There is a location $j$ where $Q_j\in\{X,Y,Z\}$ and $\tilde{Q}_j=I$.
In this case $\ket{\psi_{in}}_j$ is always a $+1$-eigenvector of $\tilde{Q}_j$, but it is a $+1$-eigenstate or $-1$-eigenstate of $Q_j$ with probability $1/2$ each. Again, the $j$th qubit will contribute a uniformly random sign to the measurement outcome.
\end{enumerate}
There could be multiple $j$ where $Q_j\neq\tilde{Q}_j$;  each will add a random sign, multiplying out to one random sign. The probability that this random sign equals the value $\lambda$ that we expect to obtain as measurement outcome if $U=\tilde{U}$, is $1/2$.
Accordingly, since we have probability~$\geq 1/2$ that $Q$ and $\tilde{Q}$ differ in at least one $j$, our probability to detect a difference between $U$ and $\tilde{U}$ is~$\geq 1/4$.

\medskip

\noindent
{\bf Case 2:} There is an $R\in{\cal P}^n\setminus\{I^{\otimes n}\}$ s.t.\ conjugation by $U$ and $R\tilde{U}$ are the same maps.\\
Since $P$ is uniformly random, in each location~$j$ where $R_j\neq I$, the Paulis $R_j$ and $P_j$ at that location will commute with probability 1/2 (namely if $P_j$ is chosen to be $I$ or $R_j$) and anti-commute with probability 1/2 (namely if $P_j$ is chosen to be one of the other 2 Paulis), independently of what happens in the other locations. In the locations~$j$ where $R_j=I$, this will always commute with $P_j$. Hence $RPR=P$ with probability 1/2 and $RPR=-P$ with probability 1/2.
We know $U\ket{\psi_{in}}$ is a $\lambda$-eigenstate of $P$.
But then it will be a $-\lambda$-eigenstate of $RPR$ with probability 1/2. Hence $\tilde{U}\ket{\psi_{in}}$ will be a $-\lambda$-eigenstate of~$P$ with probability~1/2.

\medskip

In sum, our test will output the known value $\lambda\in\{+1,-1\}$ with probability~1 if $U=\tilde{U}$, but will output $-\lambda$ with probability at least~$1/4$  if $U\neq \tilde{U}$. This allows us to detect that $U$ and $\tilde{U}$ are different Cliffords. 

The cost of this test is essentially as small as could be:
computing $Q=U^\dagger(P)$ has classical cost linear in the size of the known circuit~$U$; then we need to prepare the $n$-qubit tensor-product state $\ket{\psi_{in}}$, run $\tilde{U}$ once on it, and measure $P$ on the resulting state.
This gives us constant probability of detecting a difference between the two Clifford circuits $U$ and $\tilde{U}$ if there is one.
Note that $n$ single-qubit Pauli measurements according to~$P=P_1\otimes\cdots\otimes P_n$ would also suffice: the expectation value of (indeed, the whole distribution of) the product of the $n$ single-qubit measurement outcomes  is the same as that of~$P$.  This might be easier to realize technologically than one overall $\pm 1$-valued $n$-qubit measurement. 

Running our
test $k$ times, with fresh random~$P$ in each run, will detect $U\neq\tilde{U}$ with success probability $\geq 1-(3/4)^k$.
Setting $k=\ceil{\log(1/\delta)/\log(4/3)}$, the detection probability is $\geq 1-\delta$.
If we fix $\delta$ to some small constant, then we need to run our test only a constant number of times in order to achieve such high success probability.
As we noted in the introduction, our test still works to detect whether the implemented Clifford circuit equals $U$ or not, even if the errors are different in each run (i.e., if $\tilde{U}$ is a different erroneous Clifford in different runs).

Our ability to detect a difference from $U$ with probability as close as we wish to~1 using $O(1)$ runs of $\tilde{U}$, also means that a small (but constant) additional error probability due to the unavoidable noise and decoherence in each of these runs still leaves us with high success probability.

\subsection{Finding the error(s)}

The previous section gave a test to see whether $n$-qubit Clifford circuit $U$ (of which we have a classical description) equals another $n$-qubit Clifford circuit $\tilde{U}$ (which we can run as a black-box) or differs from it in some way. If we are in the latter situation, it would be nice if we can efficiently find out where and what the difference was.

Using a number of runs of the above test, we can indeed identify the error, or at least something equivalent to it.
The idea is the following: the known circuit $U$ acts on $n$ qubits and has $s$ gates, so the number of circuits $U'$ that differ from $U$ in one gate (or one Pauli error) is relatively small, only $O(s)$. Accordingly, we can just run the above test for each of those $U'$, testing whether the known Clifford circuit $U'$ equals the circuit $\tilde{U}$ (which we can still run as a black-box). 

Note that the same idea also works if there can be up to $d$ gate-differences  instead of one. However, the number of circuits $U'$ that are within $d$ errors of $U$ is roughly $s^d$, so the number of tests grows quickly (though still polynomially if $d=O(1)$).  Having learnt $\tilde{U}$, we can correct it.

\subsection{Deriving the same Clifford-testing result from \cite{flammia&liu:fidelityestimation} and \cite{SLP:practical}}\label{ssec:flammialiu}

As mentioned in the introduction, after finishing our Clifford test of Section~\ref{ssec:cliffordtest}, we discovered that something very similar can be derived from work of Flammia and Liu~\cite{flammia&liu:fidelityestimation} and da Silva, Landon-Cardinal, and Poulin~\cite{SLP:practical}.
Specifically, Flammia and Liu~\cite{flammia&liu:fidelityestimation} describe a procedure that, given the classical description of a Clifford circuit~$U$ and the ability to run another quantum operation~$\tilde{U}$ as a black-box, estimates (with success probability $\geq 1-\delta$) their entanglement fidelity up to additive error $\leq\eps$ using $O(\frac{1}{\eps^2}\log(1/\delta))$ runs of $\tilde{U}$. Very similar to ours, each run of $\tilde{U}$ in their procedure starts with a product state of eigenstates of a random Pauli, and ends with a Pauli measurement on the final state.
In fact, their procedure even works if $\tilde{U}$ is a general quantum channel (CPTP map) rather than a unitary.

For general unitary circuits, the entanglement fidelity $|\frac{1}{2^n}\Tr(U^\dagger \tilde{U})|^2$ can be arbitrarily close to~1, which means one has to have arbitrarily small $\eps$ to ``see'' the difference between the case $U=\tilde{U}$ and the case where $U$ and $\tilde{U}$ are distinct but have a lot of overlap. However, in the special case where $U$ and $\tilde{U}$ are distinct Clifford circuits, we show below that the entanglement fidelity  is at most $1/2$. Hence running the Flammia-Liu procedure with constant $\eps<1/4$ suffices to detect (with probability $\geq 1-\delta$) any difference between Clifford circuits $U$ and $\tilde{U}$, using just  $O(\log(1/\delta))$ runs of $\tilde{U}$ on product-state inputs and with Pauli measurement at the end, just like our test.

\begin{theorem}
If $U$ and $\tilde{U}$ are distinct $n$-qubit Clifford unitaries,
then $\displaystyle\left|\frac{1}{2^n}\Tr(U^\dagger\tilde{U})\right|^2\leq 1/2$.
\end{theorem}

\begin{proof}
It suffices to prove that $|\Tr(U)|^2\leq 2^{2n-1}$ for every non-identity Clifford~$U$.
If $U=U_1\otimes\cdots\otimes U_n$ is a product of Paulis then $\Tr(U)=\prod_{j=1}^n\Tr(U_j)=0$, because at least one of the $U_j$'s must be $X$, $Y$ or $Z$, which have trace~0.

If, on the other hand,  $U$ is not a product of Paulis, then by Theorem~\ref{th:halfthe paulis}, conjugation by $U$ maps at least half of all $P\in{\cal P}^n$ to $\pm P'$ for some $P'\neq P$.

Let $\ket{\psi}=\frac{1}{\sqrt{2^n}}\sum_{i\in\01^n}\ket{i}\ket{i}$ be the $2n$-qubit maximally entangled state. 
It is well known (and easy to verify) that for all $2^n$-dimensional matrices $A$ and $B$, we have
\[
\bra{\psi}(A\otimes B)\ket{\psi}=\frac{1}{2^n}\Tr(A^T B).
\]
Note that the $2^{2n}$ states $(I\otimes P)\ket{\psi}$, $P\in{\cal P}^n$, form an orthonormal set,
hence
\[
I_{2^{2n}}=\sum_{P\in{\cal P}^n}(I\otimes P)\ketbra{\psi}{\psi}(I\otimes P).
\]
Let $\bar{U}$ denote the entrywise conjugate of $U$ (without transposition, so $\bar{U}^T=U^\dagger$). Repeatedly using cyclicity of trace, we can now write
\begin{align*}
|\Tr(U)|^2 & = \Tr(\bar{U}\otimes U)\\
& = \Tr\left((\bar{U}\otimes U)\cdot \sum_{P\in{\cal P}^n}(I\otimes P)\ketbra{\psi}{\psi}(I\otimes P)\right)\\
& = \sum_{P\in{\cal P}^n}\Tr\left((\bar{U}\otimes PUP) \ketbra{\psi}{\psi}\right)\\
&=\sum_{P\in{\cal P}^n}\bra{\psi}(\bar{U}\otimes PUP) \ket{\psi}\\
&=\frac{1}{2^n}\sum_{P\in{\cal P}^n}\Tr(U^\dagger PUP)\\
&=\frac{1}{2^n}\sum_{P\in{\cal P}^n}\Tr(P\cdot UPU^\dagger).
\end{align*}
For at least half of all $P\in{\cal P}^n$, $UPU^\dagger$ is $\pm P'$ for some $P'\neq P$, in which case $\Tr(P\cdot UPU^\dagger)=0$. For the other $P\in{\cal P}^n$ (of which there are at most $\frac{1}{2}2^{2n}$), where $UPU^\dagger=\pm P$, we have $\Tr(P\cdot UPU^\dagger)=\pm 2^n$. 
Hence we obtain our desired upper bound:
\[
|\Tr(U)|^2\leq \frac{1}{2^n}\frac{1}{2}2^{2n}2^n=2^{2n-1}.
\]
(Note that this bound is exactly tight for $U=S\otimes I\otimes\cdots \otimes I$, since $|\Tr(S)|^2=2$.)
\end{proof}

\section{Future work}

The goal of lightweight testing and verification of quantum circuits is an important one, especially considering the severe limitations of medium-term quantum computing hardware. In this paper we gave several examples of non-trivial tests one can do to efficiently check whether two circuits are equal or differ in a worst-case distance measure, and in some cases to find the error. Our tests are far from satisfactory, though, and we hope they can be improved in various directions. Below we mention some questions for future work:
\begin{itemize}
\item {\bf Simpler tests.} Can we design better tests that are more lightweight? In particular, the preparation of $2n$ EPR-pairs in Section~\ref{sec:generalgates}, and the preservation of entanglement among those qubits for the duration of the test, is hard to realize in experiments. Can we do something like this with much less entanglement? (see footnote~\ref{jonasfootnote} for one approach)
\item {\bf More general errors.} In Section~\ref{sec:clifford} we handled the situation where our Clifford circuit $U$ is implemented as a circuit $\tilde{U}$ which may be wrong, but is assumed still to be Clifford. However, errors can be of many types. What about testing for a Clifford circuit with one arbitrary \emph{unitary} but possibly non-Clifford error~$V$?
Such a $V$ can be written as a linear combination of the Paulis, so something should be possible along the lines of this paper, but we have not worked this out yet.
Of course, an even more general setting would be arbitrary not-even-unitary errors on some of the qubits, which correspond to arbitrary CPTP maps; in this case we should aim at detecting a large distance in something like the ``diamond norm'' rather than the $D^{max}$-norm.
\item While our Clifford test of Section~\ref{sec:clifford} does not care whether there are one or more faulty gates, the test for general circuits of Section~\ref{sec:generalgates} does. As we showed in Section~\ref{ssec:twofaults}, the close relation between the average-case $D$-distance between two circuits (which is what we can test for) and their worst-case $D^{max}$-distance (which is what we would like to test for) already disappears when we have two faulty gates instead of one.
How can we detect the presence of {\bf multiple faulty gates} in the general, non-Clifford situation?
\item In some cases one can conjugate a possibly-faulty gate with random gates in order to {\bf convert adversarial noise to random noise} (see e.g.\ the work of Wallman and Emerson~\cite{wallman&emerson:noisetailoring}). Can we use that somehow? Such an approach might help bridge the gap between average-case and worst-case distance measures.
\end{itemize}

\subsection*{Acknowledgements.}
We thank Harry Buhrman, Richard Jozsa and Ashley Montanaro for useful interactions in the early stages of this work, in particular Richard for his inspiring note~\cite{jozsa:cliffordnote}. Also many thanks to Ashley, Sergio Boixo, Steve Flammia, Jonas Helsen, Yi-Kai Liu, and the anonymous Quantum referees for very helpful comments and pointers to related work.

NL was partially supported by the UK Engineering and Physical Sciences Research Council through grants EP/R043957/1, EP/S005021/1, EP/T001062/1.
RdW was partially supported by ERC Consolidator Grant 615307-QPROGRESS (which ended February 2019), and by the Dutch Research Council (NWO) through Gravitation-grant Quantum Software Consortium, 024.003.037, and through QuantERA ERA-NET Cofund project QuantAlgo 680-91-034.

\bibliographystyle{alphaUrlePrint}
\bibliography{qc}

\end{document}